\documentclass{llncs}
\usepackage{makeidx}
\usepackage{graphicx}             
\usepackage{amsmath}       
\usepackage[latin1]{inputenc}  
\usepackage{subfigure}
\newcommand{\RR}{\text{R}}      
\newcommand{\ZZ}{\text{Z}}

\newcommand{\lc}{\left\{}  
\newcommand{\rc}{\right\}}

\newcommand{\lb}{\left(}  
\newcommand{\rb}{\right)}

\begin{document}

\title{Network bargaining with general capacities}

\author{Linda Farczadi\inst{1},  Konstantinos Georgiou \inst{1}, Jochen K\"{o}nemann \inst{1}}

\institute{University of Waterloo, Waterloo, Canada}

\maketitle

\begin{abstract}
We study \emph{balanced solutions} for \emph{network bargaining games} with general capacities, where agents can participate in a fixed but arbitrary number of contracts. We provide the first polynomial time algorithm for computing balanced solutions for these games. In addition, we prove that an instance has a balanced solution if and only if it has a stable one. Our methods use a new idea of reducing an instance with general capacities to a network bargaining game with unit capacities defined on an auxiliary graph. This represents a departure from previous approaches, which rely on computing an allocation in the intersection of the \emph{core} and \emph{prekernel} of a corresponding \emph{cooperative game}, and then proving that the solution corresponding to this allocation is balanced. In fact, we show that such cooperative game methods do not extend to general capacity games, since  contrary to the case of unit capacities, there exist allocations in the intersection of the core and prekernel with no corresponding balanced solution. Finally, we identify two sufficient conditions under which the set of balanced solutions corresponds to the intersection of the \emph{core} and \emph{prekernel}, thereby extending the class of games for which this result was previously  known.
\end{abstract}

\section{Introduction}

Exchanges in networks have been studied for a long time in both sociology and economics. In sociology, they appear under the name of \emph{network exchange theory}, a field which studies the behaviour of agents who interact across a network to form bilateral relationships of mutual benefit. The goal is to determine how an agent's location in the network influences its ability to negotiate for resources \cite{cook1992power}. In economics, they are known as \emph{cooperative games} and have been used for studying the distribution of resources across a network, for example in the case of  two-sided markets \cite{shapley1971assignment} \cite{rochford1984symmetrically}.

From a theoretical perspective the most commonly used framework for studying such exchanges is that of \emph{network bargaining games}. The model consists of an undirected graph $G=(V,E)$ with edge weights $w: E(G) \rightarrow \RR_{+}$ and vertex capacities $c: V(G) \rightarrow \ZZ_+$. The vertices represent the agents, and the edges represent possible pairwise contracts that the agents can form.  The weight of each edge represents the value of the corresponding contract. If a contract is formed between two vertices, its value is divided between them,  whereas if the contract is not formed neither vertex receives any profit from this specific contract. The capacity of each agent limits the number of contracts it can form. This constraint, together with an agent's position in the network determine its bargaining power. 

A \emph{solution} for the network bargaining model specifies the set of contracts which are formed, and how each contract is divided. Specifically, a solution consists of a pair $(M,z)$, where $M$ is a $c$-matching of the underlying graph $G$, and $z$ is a vector which assigns each edge $uv$  two values $z_{uv},z_{vu} \geq 0$ corresponding to the profit that agent $u$, respectively agent $v$, earn from the contract $uv$. To be a valid solution, the two values $z_{uv}$ and $z_{vu}$ must add up to the value of the contract whenever the edge $uv$ belongs to the $c$-matching $M$, and must be zero otherwise. 

Solutions to network bargaining games are classified according to two main concepts: \emph{stability} and \emph{balance}. A solution is stable if the profit an agent earns from any formed contract is at least as much as its outside option. An agent's \emph{outside option}, in this context, refers to the  maximum  profit that the agent can rationally receive by forming a new contract with one of its neighbours, under the condition that the newly formed contract would benefit both parties. The notion of balance, first introduced in \cite{cook1983distribution}, \cite{rochford1984symmetrically}, is a generalization of the Nash bargaining solution to the network setting.  Specifically, in a balanced solution the value of each contract is split according to the following rule: both endpoints must earn their outside options, and any surplus is to be divided equally among them. Balanced solutions have been shown to agree with experimental evidence, even to the point of picking up on subtle differences in bargaining power among agents \cite{willer1999network}. This is an affirmation of the fact that these solutions are natural and represent an important area of study. 

There is a close connection between network bargaining games and \emph{cooperative games}. Specifically given a solution $(M,z)$  to the network bargaining game $(G,w,c)$ we can define a corresponding payoff vector $x$, where $x_u$ is just the total profit earned by vertex $u$ from all its contracts in the solution $(M,z)$. Then this vector $x$ can be seen as a solution to a corresponding  cooperative game $(N,v)$ defined as follows: we let $N = V(G)$ denote the set of players, and for every subset $S \subseteq N$ of players,  we define its \emph{value $\nu(S)$} as the weight of the maximum weight $c$-matching of $G[S]$. This game is also known as the \emph{matching game}  \cite{shapley1971assignment}. The subsets $S \subseteq N$ are referred to as \emph{coalitions}, and the value $\nu(S)$ of each coalition is interpreted as the collective payoff that the players in $S$ would receive if they decide to cooperate. The players are assumed to be able choose which coalitions to form, and their objective is to maximize their payoffs. 

The underlying assumption in cooperative game theory is that the grand coalition $N$ will form, and the question becomes how to distribute the payoff $\nu(N)$ among the players. A vector $x$  describing such a distribution is referred to as an \emph{allocation}. Given an allocation $x$, the excess of a coalition $S$ is defined as $\nu(S) - x(S)$. Intuitively, a negative excess means that the players of the coalition have no incentive to break away from the grand coalition, since the collective payoff they could achieve by forming the new coalition is less then what they are currently receiving. The  \emph{power} of player $u$ over player $v$ with respect to the allocation $x$ is the maximum excess achieved by a coalition which includes  $u$ but excludes $v$. Two important concepts from cooperative games are those of the \emph{core} and \emph{prekernel}. The core consists of  allocations for which no coalition has a negative excess, whereas the prekernel consists of all allocations with symmetric powers.

\textbf{Our contribution and results.} Our main result is providing the first polynomial time algorithm for computing balanced solutions for network bargaining games with general capacities and fully characterizating the  existence of balanced solutions for these games. Specifically we show the following results in sections \ref{sec1} and \ref{sec2} respectively:

\textsc{Result 1.} \textit{There exists a polynomial time algorithm which given an instance of a network bargaining game with general capacities and a maximum weight $c$-matching $M$,  computes a balanced solution  $(M,z)$ whenever one exists.}

\textsc{Result 2.} \textit{A network bargaining game with general capacities has a balanced solution if and only if it has a stable one.}

Our method relies on a new approach of reducing a general capacity instance to a network bargaining game with unit capacities defined on an auxiliary graph. This allows us to use existing algorithms for obtaining balanced solutions for unit capacities games, which we can then transform to  balanced solutions of our original instance. This represents a departure from previous approaches of \cite{bateni2010cooperative} which relied on proving an equivalence between the set of balanced solutions and the intersection of the core and prekernel of the corresponding matching game. In section  \ref{sec3} we show that such an approach cannot work for our case, since this equivalence does not extend to all instances of general capacity games:

\textsc{Result 3.} \textit{There exists an instance of a network bargaining game with general capacities for which we can find an allocation in the intersection of the core and  prekernel such that there is no corresponding balanced solution for this allocation}.

Despite this result, we provide two necessary conditions which ensure that the correspondence between the set of balanced solutions and allocations in the intersection of the core and prekernel is maintained. Using the definition of \emph{gadgets} from section \ref{gadgets} we have the following result given in  section \ref{sec4}:

\textsc{Result 4.} \textit{If the network bargaining game has no gadgets and the maximum $c$-matching $M$ is acyclic, the set of balanced solutions corresponds to the intersection of the core and prekernel.}

\textbf{Related work.} Kleinberg and Tardos \cite{kleinberg2008balanced} studied network bargaining games with unit capacities and developed a polynomial time algorithm for computing the entire set of balanced solutions. They also show that such games have a balanced solution whenever they have a stable one and that a stable solution exists if and only if the linear program for the maximum weight matching of the underlying graph has an integral optimal solution. 

Bateni et al. \cite{bateni2010cooperative} consider network bargaining games with unit capacities, as well as the special case of network bargaining games on bipartite graphs where one side of the partition has all unit capacities. They approach the problem of computing balanced solutions from the perspective of cooperative games. In particular they use the matching game of Shapley and Shubik \cite{shapley1971assignment} and show that the set of stable solutions corresponds to the core, and the set of balanced solutions corresponds to the  intersection of the core and prekernel.  

Like we do here, Kanoria et al. \cite{kanoria2009natural} also study network bargaining games with general capacities. They show that a stable solution exists for these games if and only the linear program for the maximum weight $c$-matching of the underlying graph has an integral optimal solution. They are also able to obtain a partial characterization of the existence of balanced solutions by proving  that if this integral optimum is unique, then a balanced solution is guaranteed to exist. They provide an algorithm for computing balanced solutions in this case which uses local dynamics but whose running time is not polynomial.

\section{Preliminaries and definitions}

An instance of  the \emph{network bargaining game}  is a triple $\lb G, w,c \rb$ where $G $ is a undirected graph,  $w \in \RR_{+}^{|E(G)|}$ is a vector of edge weights, and  $c  \in \ZZ_{+}^{|V(G)|}$ is a vector of vertex capacities. A set of edges $M \subseteq E(G)$ is a $c$-matching of $G$ if  $| \lc v: uv \in M \rc |\leq c_u$ for all $u \in V(G)$. Given a $c$-matching $M$, we let $d_u$ denote the degree of vertex $u$ in $M$. We say that vertex $u$ is \emph{saturated} in $M$ if $d_u = c_u$. 

A \emph{solution} to the network bargaining game $(G,w,c)$ is a pair $\lb M,z \rb$ where $M$ is a $c$-matching of $G$ and $z \in \RR_+^{2|E(G)|}$ assigns each edge $uv$ a pair of values $z_{uv},z_{vu}$ such that  $z_{uv}+z_{vu} = w_{uv} $ if $uv \in M$ and $z_{uv} = z_{vu} = 0$ otherwise. 

The \emph{allocation} associated with the solution $(M,z)$ is the vector $x \in \RR^{|V(G)|}$ where $x_u$ represents the total payoff of vertex $u$, that is for all $u \in V(G)$ we have $x_u = \sum_{v: uv \in M} z_{uv}$. The \emph{outside option} of  vertex $u$ with respect to a solution $(M,z)$ is defined as
\begin{align}
\alpha_u(M,z):= \max \lb 0,  \max_{v: uv \in E(G) \backslash M}  \lb w_{uv}-\textbf{1}_{[d_v= c_v]} \min_{vw \in M} z_{vw} \rb \rb, \notag
\end{align}
where $\textbf{1}_{E}$ is the indicator function for the event $E$, which takes value one whenever the event holds, and zero otherwise. If $ \lc v : uv \in  E(G)  \backslash M \rc = \emptyset$ then we set $\alpha_u(M,z) = 0$. We write $\alpha_u$ instead of $\alpha_u(M,z)$ whenever the context is clear.

A solution $(M,z)$ is \emph{stable} if for all $uv \in M$ we have $z_{uv} \geq \alpha_{u}(M,z)$, and for all unsaturated vertices $u$ we have $\alpha_{u}(M,z) =0$.

A solution $(M,z)$ is \emph{balanced} if it is stable and in addition for all $uv \in M$ we have $z_{uv} - \alpha_{u}(M,z) = z_{vu} - \alpha_{v}(M,z)$.

\subsection{Special case: unit capacities}

The definitions from the previous section simplify in the case where all vertices have unit capacities. Specifically a \emph{solution} to the unit capacity game $(G,w)$ is a pair $\lb M,x \rb$ where $M$ is now a matching of $G$ and $x \in \RR_{+}^{|V(G)|}$ assigns a value to each vertex such that for all edges $uv \in M$ we have $x_{u} +x_{v} = w_{uv}$  and for all $u \in V(G)$ not covered by $M$ we have $x_u= 0$. Since each vertex has at most one unique contract, the vector $x$ from the solution $(M,x)$ is also the allocation vector in this case. 

The \emph{outside option} of  vertex $u$ can now be expressed as
\begin{align}
\alpha_u(M,x):=  \max_{v: uv \in E(G) \backslash M} \lb   w_{uv}-x_{v}\rb , \notag
\end{align}
where as before we set $\alpha_u(M,x) = 0$ whenever $ \lc v : uv \in E(G) \backslash M \rc = \emptyset$. 

A solution $(M,x)$ is \emph{stable} if for all $u \in V(G)$ we have $x_{u} \geq \alpha_{u}(M,x)$ and  \emph{balanced} if it is stable and in addition  $x_{u} - \alpha_{u}(M,x) = x_{v} - \alpha_{v}(M,x)$ for all $uv \in M$.

\subsection{Cooperative games}

Given an instance $(G,w,c)$ of the network bargaining game we let $N=V(G)$ and define the \emph{value} $\nu(S)$ of a set of vertices $S \subseteq N$ as
\begin{align}
\nu(S) := \max_{\text{$M$ $c$-matching of $G[S]$}} w(M). \notag
\end{align}
Then the pair $(N,v)$ denotes an instance of the \emph{matching game} of Shapley and Shubik \cite{shapley1971assignment}. We will refer to this as the matching game associated with the instance $(G,w,c)$. 

Given  $x \in \RR_{+}^{|N|}$ and two  vertices $u, v \in V(G)$ we define the \emph{power} of vertex $u$ over vertex $v$ with respect to the vector $x$ as
\begin{align}
s_{uv}(x) &:= \max_{S \subseteq N: u \in S, v \notin S} \nu(S) - x(S), \notag
\end{align}
where $x(S) = \sum_{u \in S}x_u$. We write $s_{uv}$ instead of $s_{uv}(x)$ whenever the context is clear. The \emph{core} of the game is defined as the set
\begin{align}
\mathcal{C} := \lc x \in \RR_{+}^{|N|} : x(S) \geq \nu(S), \, \forall \, S \subset N, x(N) = \nu(N) \rc. \notag
\end{align}
The \emph{prekernel} of the game is the set
\begin{align}
\mathcal{K} := \lc  x \in \RR_{+}^{|N|}: s_{uv}(x) = s_{vu}(x) \quad \forall \, u,v \in N \rc \notag.
\end{align}

\section{Balanced solutions via cooperative games}

The first attempt towards computing balanced solutions for the network bargaining game with general capacities is to use the connection to cooperative games presented in \cite{bateni2010cooperative}. For the special class of unit capacity and constrained bipartite games, Bateni et al. show that the set of stable solutions corresponds to the core, and the set of balanced solutions  corresponds to the intersection of the core and prekernel of the associated matching game. This implies that efficient algorithms, such as the one of \cite{faigle1998efficient}, can be used to compute points in the intersection of the core and prekernel from which a balanced solution can be uniquely obtained.

\subsection{Allocations in $\mathcal{C} \cap \mathcal{K}$ with no corresponding balanced solutions}\label{sec3}

The first question of interest is whether this equivalence between balanced solutions and the intersection of the core and prekernel extends to network bargaining games with arbitrary capacities. The following lemma proves that this is not always the case. 

\begin{lemma}\label{lem: 1} There exists an instance $(G,w,c)$ of the network bargaining game  and a vector $x \in \mathcal{C} \cap \mathcal{K}$ such that there exists no balanced solution $(M,z)$  satisfying $x_u = \sum_{v: uv \in M} z_{uv} \quad \text{ for all $u \in V(G)$}$.
\end{lemma}

\begin{proof}Consider the following graph where every vertex has capacity 2 and the edge weights are given above each edge
\begin{center}
\includegraphics[scale=0.4]{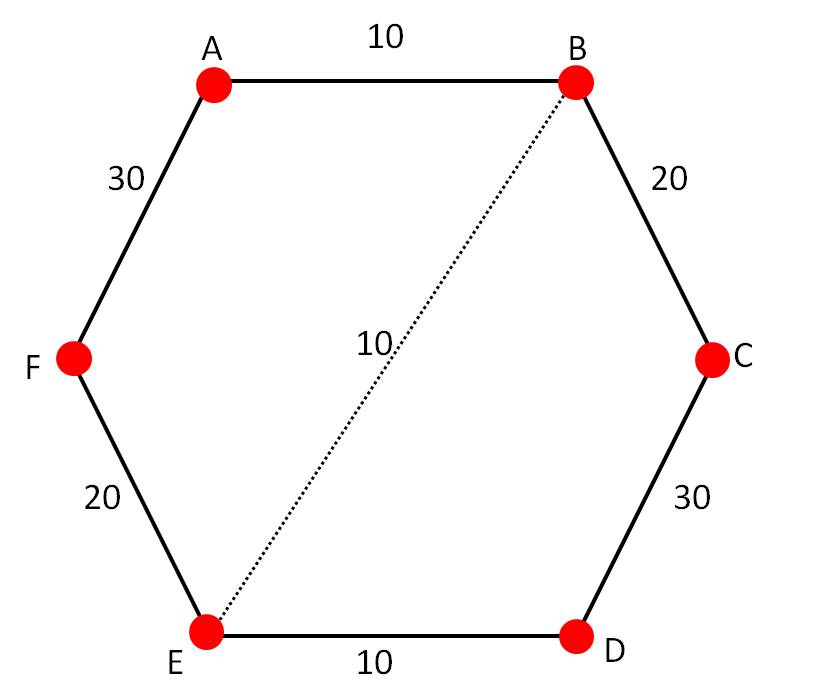}
\end{center}

Consider the vector $x$ defined as  $x_u = 20$ for all $u \in V(G)$.  We now show that the vector $x$ is in the intersection of the core and prekernel and there exists no balanced solution $(M,z)$  corresponding to $x$.

Let $C_1$ denote the outer cycle on vertices $A,B,C,D,E,F$  and let $C_2$ and $C_3$ denote the inner cycles on vertices $B,C,D,E$ and  $E,F,A,B$ respectively. The unique optimal $2$-matching is the cycle $C_1$ with weight $120$. Since any stable, and therefore balanced, solution must occur on a maximum weight $c$-matching \cite{bateni2010cooperative}, any  balanced solution  $(M,z)$  will have $M= E(C_1)$.

It can be easily checked that $x \in \mathcal{C}$. To check that  $x \in  \mathcal{K}$  we  compute the powers $s_{uv}(x) := \max_{T: u \in T, v \notin T} \nu(T) -x(T)$ for all pairs of vertices $u,v \in V(G)$. For the pair $A,B$ we have:
\begin{align}
s_{AB} &= \nu \lb \lc A,F \rc \rb - x \lb \lc A, F \rc \rb = 30 - 40 = - 10 \notag\\
s_{BA} &= \nu \lb \lc C_2 \rc \rb - x \lb \lc C_2 \rc \rb = 70- 80 = - 10 \notag.
\end{align}
Similarly for the pair $B,C$ we have:
\begin{align} 
s_{BC} &= \nu \lb \lc C_3 \rc \rb - x \lb \lc C_3 \rc \rb = 70 - 80 =  - 10 \notag\\
s_{CB} &= \nu \lb \lc C,D \rc \rb  - x \lb \lc C,D \rc \rb = 30-40 = - 10 \notag.
\end{align}
And for the pair $C,D$:
\begin{align}
s_{CD} &= \nu \lb \lc C \rc \rb - x \lb \lc C \rc \rb =  -20 \notag\\
s_{DC} &= \nu \lb \lc D \rc \rb - x \lb \lc D \rc \rb = - 20 \notag.
\end{align}

Hence the pairs $(A,B), (B,C)$ and $(C,D)$ satisfy the prekernel condition. By symmetry so do $(D,E),(E,F)$ and $(F,A)$.  

Note that for any pair $u,v$ of non-adjacent vertices, one of the two cycles $C_2$ or $C_3$ will contain $u$ but not $v$, and viceversa. Therefore  $s_{uv}=s_{vu} = -10$ for all non-adjacent pairs $u,v$.  This proves that $x$ is in the prekernel.

We now show that there is no vector $z$ such that $(M,z)$ is a balanced solution corresponding to the vector $x$. First note that vertices $A,F,C$ and $D$ have an outside option of zero in any solution, since there are no edges in $E \backslash M$ incident with these vertices. Hence the contracts $(C,D)$ and $(A,F)$ have to be split evenly in any balanced solution. Since each vertex must have a total profit of $20$ from its two contracts in $M$, this uniquely determines all values of the vector $z$, which are shown in the figure below
\begin{center}
\includegraphics[scale=0.4]{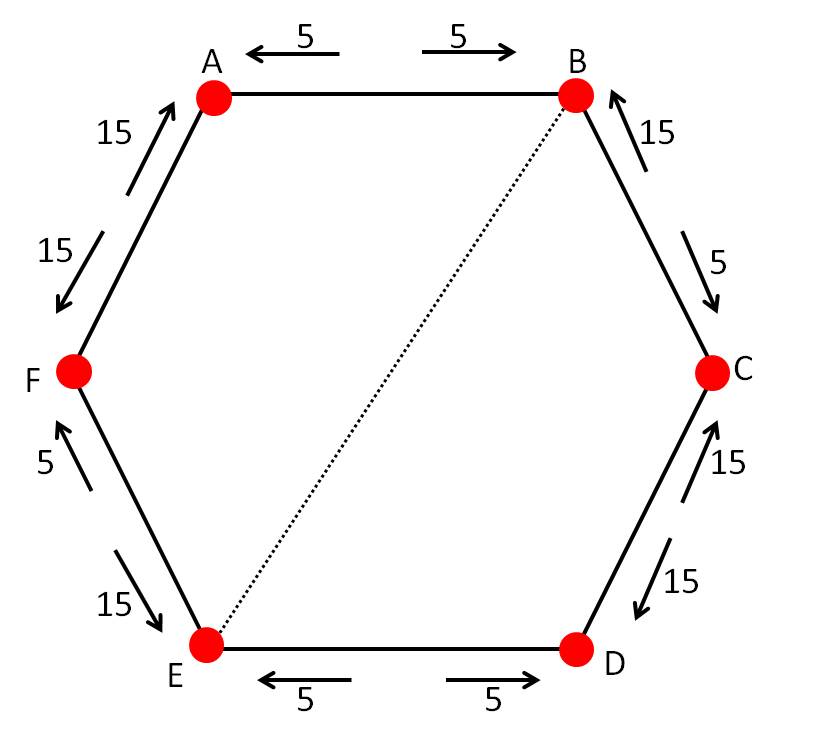}
\end{center}
The minimum contract of both $B$ and $E$ is $5$ and therefore  $\alpha_B = \alpha_E = 10 -5 = 5$. However, the edge $(B,C)$ (and also the edge $(E,F)$ by symmetry) violates the balance condition since $z_{BC} - \alpha_B = 15-5 = 10$ while $z_{CB} - \alpha_C = 5-0 = 5$.

Note that this instance does possess a balanced solution as shown in the figure below
\begin{center}
\includegraphics[scale=0.4]{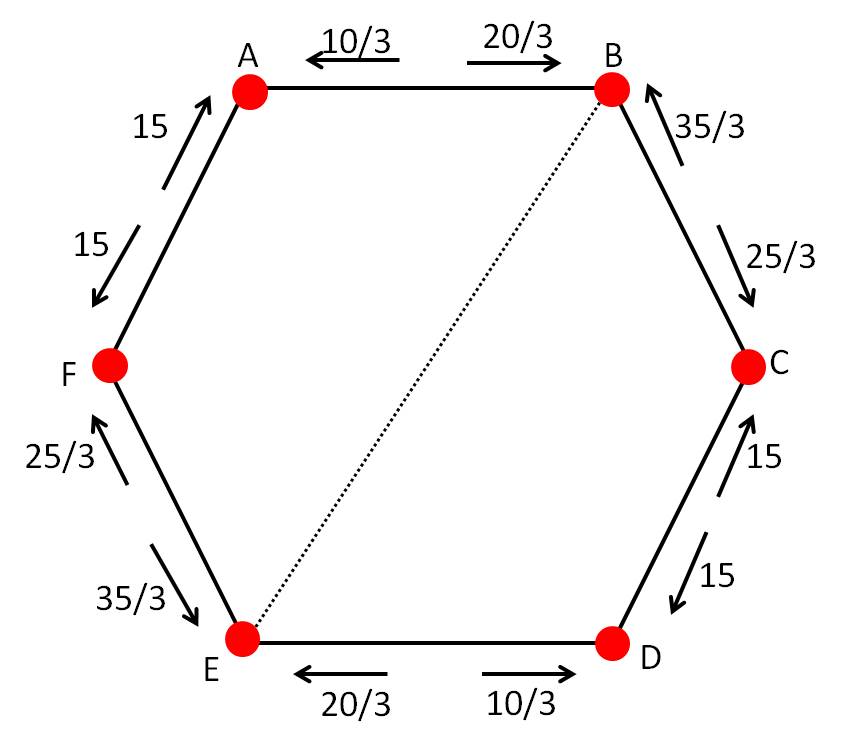}
\end{center}
Here the outside option of both $B$ and $E$ is $10/3$ and all edges in the matching satisfy the balance condition. We also remark that the allocation  associated with this balanced solution is also in the intersection of the core and prekernel. \qed
\end{proof}

In view of Lemma \ref{lem: 1}, we cannot hope to extend the correspondence between  balanced solutions and allocations in the intersection of the core and prekernel to all network bargaining games. However  we  can generalize the results of \cite{bateni2010cooperative} by characterizing a larger class of network bargaining games, including unit capacity and constrained bipartite games, for which this correspondence holds. We  achieve this by defining a certain gadget whose absence, together with the fact that the $c$-matching $M$ is acyclic, will be sufficient for the correspondence to hold.

\subsection{Gadgets}\label{gadgets}

Let $(G,w,c)$ be an instance of the network bargaining game and $(M,z)$ a solution. Consider a vertex $u \in V(G)$ with $\alpha_u(M,z) > 0$ and let $v$ be a neighbour of $u$ in $M$. Let $v'$ be vertex $u$'s best outside option and if $v'$ is saturated in $M$, let $u'$ be its weakest contract. Using these definitions we have
\begin{align}
\alpha_u(M,z) = w_{uv'}-\textbf{1}_{[d_{v'}= c_{v'}]}z_{v'u'}. \notag
\end{align}

We say that $u$ is a $\emph{bad}$ vertex in the solution $(M,z)$ if at least one of the following two conditions holds:
\begin{enumerate}
\item There is a $v-v'$ path in $M$,
\item There is a $u-u'$ path in $M$, that does not pass through vertex $v'$.
\end{enumerate}
We refer to such $v-v'$ or $u-u'$ paths as $\emph{gadgets}$ of the solution $(M,z)$. The following figure depicts these two types of gadgets, solid lines denote edges in $M$ and dashed lines denote edges in $E \backslash M$.
\begin{center}
\includegraphics[scale=0.4]{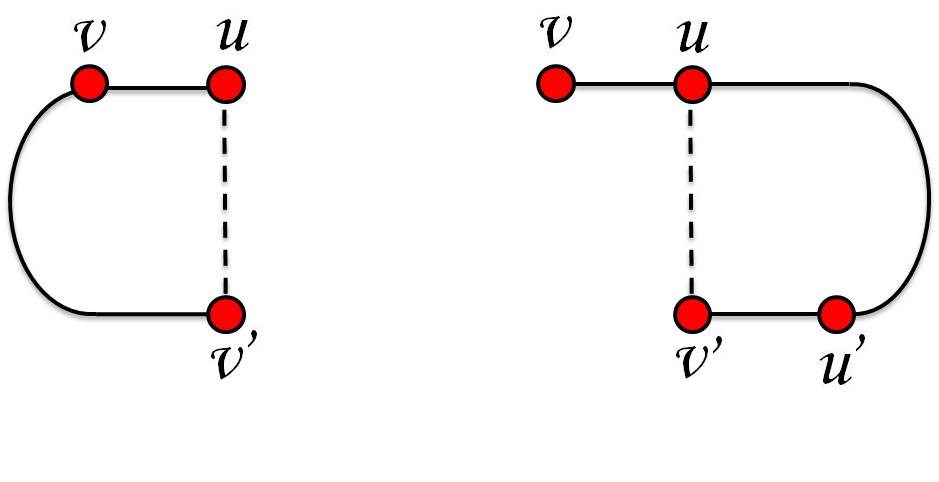}
\end{center}

\subsection{Sufficient conditions for correspondence between set of balanced solutions and $\mathcal{C} \cap \mathcal{K}$}\label{sec4}

We can now state our main theorem of this section.

\begin{theorem}\label{thm: 1} Let $(G,w,c)$ be an instance of the network bargaining game. Let $ x \in \mathcal{C}$ and $(M,z)$ be a corresponding stable solution so that $x_u = \sum_{v: uv \in M} z_{uv}$ for all $u \in V(G)$. If the following two conditions are satisfied
\begin{enumerate}
\item  $M$ is acyclic, 
\item there are no bad vertices in the solution $(M,z)$,
\end{enumerate}
then, the following statement holds
\begin{align}
\text{$(M,z)$ is a balanced solution if and only if $x \in \mathcal{K}$}. \notag
\end{align}
\end{theorem}
\begin{proof} Fix $uv \in M$. Note that it suffices to show $s_{uv} = -z_{uv} + \alpha_u$, since this would imply that $s_{uv} = s_{vu}$ if and only if $ z_{uv} - \alpha_u = z_{vu}- \alpha_v$. Our strategy is to first show that $s_{uv}$ is upper bounded $-z_{uv} + \alpha_u$, after which it will be sufficient to find a set $T$ for which $\nu(T) - x(T)$ achieves this upper bound. We start with the following lemma.

\begin{lemma}\label{lem: 2} $s_{uv} \leq -z_{uv} + \alpha_u$.
\end{lemma}

\begin{proof}
Let $T \subset N$ such that $u \in T$ and $v \notin T$. Let $M'$ be a maximum weight $c$-matching in $G[T]$. Then
\begin{align}
\nu(T) - x(T) &= w(M') - \sum_{a \in T} x_a \notag \\
&= w(M') - \sum_{a \in T, ab \in M} z_{ab} \notag \\
&= \lb w (M' \cap M ) -  \sum_{a \in T, \, ab \in M \cap M'} z_{ab} \rb  + \lb w( M' \backslash M)  - \sum_{a \in T, \, ab \in M \backslash M'} z_{ab}\rb \notag \\
&= w( M' \backslash M) - \sum_{a \in T, \, ab \in M \backslash M'} z_{ab}. \notag
\end{align}
Define the set of ordered pairs
\begin{align}
\mathcal{S} := \lc (a,b) : a \in T \cap V(G), \, b \in V(G), \, ab \in  M \backslash M' \rc, \notag 
\end{align}
so that
\begin{align}\label{key}
\nu(T) - x(T) &= w( M' \backslash M) - \sum_{(a,b) \in \mathcal{S}} z_{ab}.  
\end{align}
Since $(M,z)$ is a stable solution it follows that $M$ is a $c$-matching of maximum weight. Hence any edge in $M' \backslash M$ must have at least one saturated endpoint. Let $\lc a_1, \cdots, a_{\ell} \rc$ be the set of vertices in $T$ which are saturated in $M$. For each $i \in [\ell]$ define the sets of ordered pairs
\begin{align}
\mathcal{E}_i &:= \lc (a_i,b) : a_ib \in M' \backslash M \rc \notag\\
\mathcal{F}_i &:= \lc  (a_i,c): a_ic \in M \backslash M' \rc \notag.
\end{align}
Note that all these sets are pairwise disjoint, and $ \mathcal{F}_i \subset \mathcal{S}$ for all $i \in [\ell]$. Now since each $a_i$ is saturated in $M$  it follows that $\left|\mathcal{E}_i \right| \leq \left|\mathcal{F}_i \right|$. Therefore  we can fix an arbitrary mapping
\begin{align}
\phi_i : \mathcal{E}_i \rightarrow \mathcal{F}_i \notag
\end{align}
so that each element of $\mathcal{E}_i$ is mapped to a distinct element of  $\mathcal{F}_i$. 

We now assign to each edge  $e \in M' \backslash M$ a set of ordered elements $\mathcal{S}_e \subset \mathcal{S}$. As previously observed, each such edge $e$ must have at least one saturated endpoint. Hence let $e=a_iy$ and define:
\begin{align}
\mathcal{S}_e = \begin{cases}  \lc \phi_i  (a_i,a_j), \phi_j \lb a_j,a_i \rb \rc & \text{ if $y= a_j$ for some $j \in [\ell] \backslash \lc i \rc$,} \\
\lc \phi_i \lb a_i,y \rb \rc & \text{otherwise.} \\ \end{cases} \notag
\end{align}
It follows from the definition of the sets $\mathcal{E}_i$ and the choice of the mapping $\phi_i$ that the sets $\mathcal{S}_e$ are well defined, are pairwise disjoint, and are all subsets of $\mathcal{S}$. Let
\begin{align}
\mathcal{S'} := \mathcal{S}\backslash  \lb \cup_{e \in M' \backslash M} \mathcal{S}_e \rb, \notag 
\end{align}
Then from equation \eqref{key} and the fact that each ordered pair of $\mathcal{S}$ belongs to at most one $\mathcal{S}_e$ set we obtain
\begin{align} \label{key3}
\nu(T) - x(T) &\leq \sum_{e \in M' \backslash M} \lb w_e -  \sum_{(a,b) \in \mathcal{S}_e}z_{ab}  \rb - \sum_{(a,b) \in \mathcal{S'}}  z_{ab}. 
\end{align}
Now it follows from stability that for all $e \in M' \backslash M$ we have
\begin{align}\label{key2}
w_{e} \leq \sum_{(a,b) \in \mathcal{S}_e} z_{ab}.
\end{align}
To see this, consider $e=a_iy \in  M' \backslash M$. If $y =a_j$ for some $j \in [\ell] \backslash \lc i \rc$ then $\mathcal{S}_e :=  \lc \phi_i  (a_i,a_j), \phi_j(a_j,a_i) \rc$. Now $\phi_i  (a_i,a_j)$  represents the profit that $a_i$ gets from one of his contracts in $M$, and similarly  $\phi_j \lb a_j,a_i \rb$ represents the profit that $a_j$ earns from one of his contracts in $M$. Since the edge $a_ia_j$ is not in $M$, by stability we must have $\phi_i  (a_i,a_j) + \phi_j(a_j,a_i) \geq w_{a_iy},$. 

In the other case where $y$ is not a saturated vertex in $T$, we have $\mathcal{S}_e = \lc \phi_i \lb (a_i,y) \rb \rc $. Since $y$ is an outside option for $a_i$ and  $\phi_i \lb a_i,y \rb$ represents the profit that $a_i$ gets from one of his contracts in $M$, stability for vertex $a_i$ implies that $\phi_i \lb a_i,y \rb \geq w_{a_iy}$ as required. 

Now suppose that $(u,v) \notin \cup_{e \in M' \backslash M} \mathcal{S}_e$. Then from \eqref{key3} and \eqref{key2} we have
\begin{align}
\nu(T) - x(T) &\leq -z_{uv}.\notag
\end{align}
Since $\alpha_u \geq 0$, this proves the lemma in this case. If on the other hand, there exists a set $S_{e^*}$ such that $(u,v) \in S_{e^*}$. Then using \eqref{key3} and \eqref{key2} again we have
\begin{align}
\nu(T) - x(T) &\leq - z_{uv} + w_{e^*} -  \sum_{(a,b) \in \mathcal{S}_{e^*} \backslash \lc (u,v) \rc} z_{uv}, \notag
\end{align}
Now $(u,v) \in S_{e^*}$ implies that $e^*$ must be an edge in $M' \backslash M$ that is incident to vertex $u$. Hence $e^*=uw$ where $w$ is an outside option for vertex $u$. If $w$ is not saturated then the set  $\mathcal{S}_{e^*} \backslash \lc (uv) \rc$ is  empty. Otherwise if $w$ is saturated,  this set contains a unique ordered pair $(w,k)$ such that $wk \in M \backslash M'$. Therefore it follows from the definition of $\alpha_u$ that 
\begin{align}
\nu(T) - x(T) &\leq - z_{uv} + \alpha_u, \notag
\end{align}
as desired.\qed
\end{proof}

Hence it suffices to find a set $T \subseteq V(G)$ such that $u \in T$, $v \notin T$ and show that $\nu(T) - x(T) \geq -z_{uv} + \alpha_u$. Given a set of vertices $S$ we let $M_S$ denote the edges of $M$ which have both endpoints in $S$. Note that for any set of vertices $S$ we have
\begin{align}\label{comp}
w(M_S) - x(S) = - \sum_{ab \in M : a \in S, b \notin S} z_{ab}.
\end{align}

We define $\mathcal{C}$ to be the set of components of $G$ induced by the edges in $M$. Since $u$ and $v$ are neighbours in $M$ they will be in the same component, call it $C$. Now suppose we remove the edge $uv$ from $C$. Since $M$ is acyclic, this disconnects $C$ into two components $C_u$ and $C_v$, containing vertices $u$ and $v$ respectively. Now $M_{C_u}$ is a valid $c$-matching of $C_u$ hence applying equation \eqref{comp} to the vertex set the component $C_u$ we obtain
\begin{align}
\nu(C_u) - x(C_u) &\geq w(M_{C_u}) - x(C_u) =  - z_{uv}. \notag
\end{align}

If $\alpha_u = 0$ then setting $T$ to be the vertex set of component $C_u$ completes the proof for this case. Hence it remains to consider the case where $\alpha_u > 0$. Then by stability of the solution $(M,z)$  vertex $u$ must be saturated in $M$. Let $v'$ be vertex $u$'s best outside option.

\textbf{Case 1}: $v' \in C_u$ and $v'$ is  not saturated in $M$. Since $uv \notin C_u$ and $v'$ is not saturated in $M$ the set of edges  $M_{C_u} \cup \lc uv' \rc$ is a valid $c$-matching of $C_u$ and therefore
\begin{align}
\nu(C_u) - x(C_u) &\geq w(M_{C_u} \cup \lc uv' \rc) - x(C_u) \notag\\
&= w(M_{C_u}) - x(C_u)+ w_{uv'}  \notag \\
&= - z_{uv} + w_{uv'} &&  \text{applying \eqref{comp} to $C_u$}. \notag
\end{align}

\textbf{Case 2}:  $v' \in C_u$ and $v'$ is saturated in $M$. Let $u'$ be the weakest contract of $v'$ and suppose we remove the edge $v'u'$ from $C_u$. Since $M$ is acyclic, this disconnects $C_u$ into two components. From condition $(2)$ we know that $u'$ is not on the $u-v'$ path in $M$. Hence $u$ and $v'$ are in the same component of $C_u \backslash \lc v'u' \rc$. Denote this component by $D_u$. Now  $M_{D_{u}} \cup \lc uv' \rc$ is a $c$-matching of $D_u$ and thus
\begin{align}
\nu(D_u) - x(D_u) &\geq w \lb M_{D_u} \cup \lc uv' \rc \rb - x(D_u) \notag \\
&= w(M_{D_u}) - x(D_u) + +w_{uv'} \notag \\
&= - z_{uv} - z_{v'u''} + + w_{uv'} &&  \text{applying  \eqref{comp} to $D_u$}\notag \\ 
&= -z_{uv} + \alpha_u &&   \text{ by choice of $v'$ and $u'$.} \notag 
\end{align}

\begin{figure}\label{fig}
\hfill
\subfigure[Case 1]{\includegraphics[width=4cm]{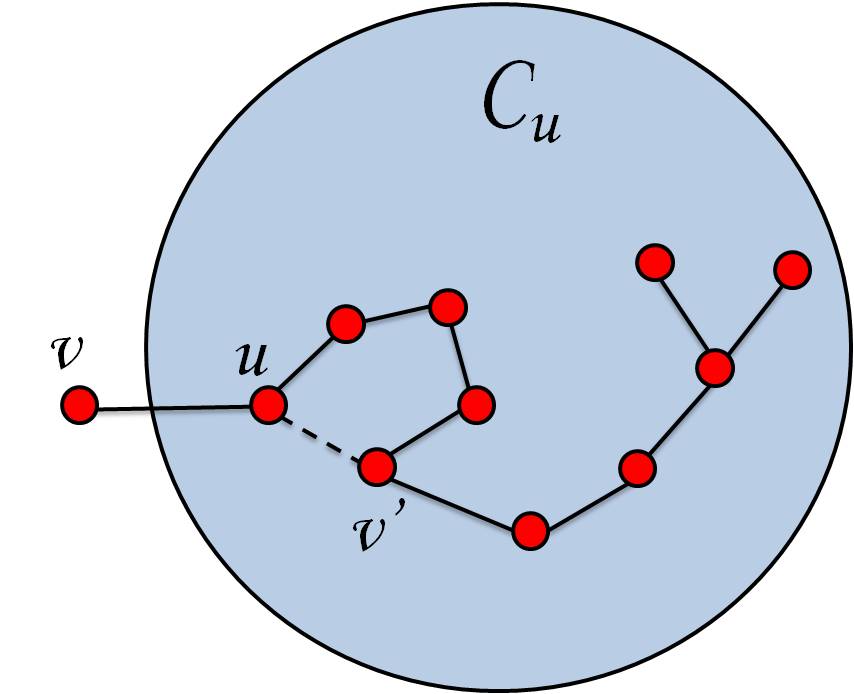}}
\hfill
\subfigure[Case 2]{\includegraphics[width=4cm]{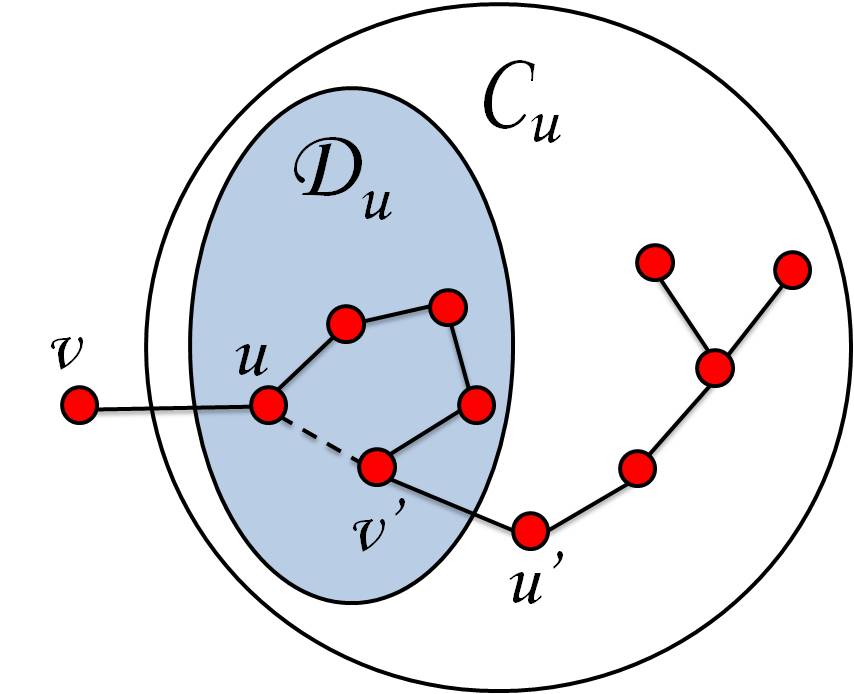}}
\hfill
\caption{Cases 1 and 2 in the proof Theorem 1}
\end{figure}

\textbf{Case 3}: $v' \notin C_u$ and $v'$ is not saturated in $M$. Condition $(2)$ implies that there is no  $v-v'$ path in $M$, hence the fact that $v' \notin C_u$  implies that $v' \notin C$. Let $C_{v'}$ be the component in $\mathcal{C}$ that contains vertex $v'$. Now $M_{C_u} \cup M_{C_{v'}} \cup \lc uv' \rc $ is a $c$-matching of $C_u \cup C_{v'}$ and therefore
\begin{align}
\nu  (C_u \cup C_{j'})  - x(C_u \cup C_{v'}) &\geq w(M_{C_u} \cup M_{C_{v'}} \cup \lc uv' \rc) - x(C_u \cup C_{v'}) \notag \\
&= \lb w(M_{C_u})- x(C_u) \rb + \lb w(M_{C_{v'}})- x(C_{v'}) \rb + w_{uv'}  \notag \\ 
&= - z_{uv} + w_{uv'} \quad \quad \quad \quad \quad  \, \text{applying  \eqref{comp} to $C_u$, $C_{v'}$} \notag \\
&= -z_{uv} + \alpha_u \quad \quad \quad \quad \quad \quad   \text{ by choice of $v'$.} \notag 
\end{align}

\begin{figure}
\begin{center}
\includegraphics[width=6.6cm]{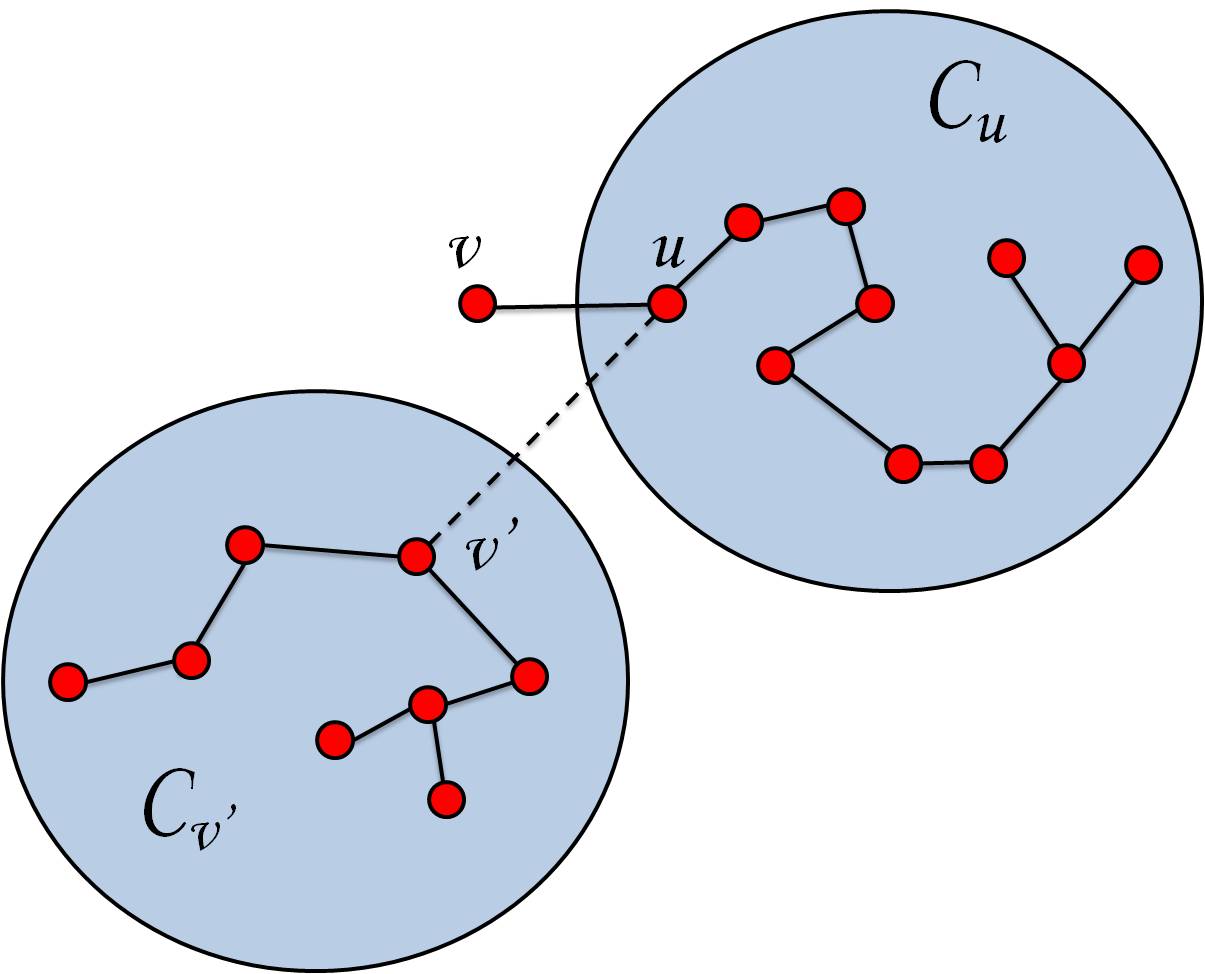}
\caption{Case 3 in the proof Theorem 1}
\end{center}
\end{figure}

\textbf{Case 4}: $v' \notin C_u$ and $v'$ is saturated in $M$. As in Case 3 we let $C_{v'}$ be the component in $\mathcal{C}$ that contains vertex $v'$. Let $u'$ be the weakest contract of $v'$ and suppose we remove the edge $v'u'$ from $C_{v'}$. Since $M$ is acyclic, this disconnects $C_{v'}$ into two components and we let $D_{v'}$ be the one that contains vertex $v'$.  Now $M_{C_u} \cup M_{D_{v'}} \cup \lc uv' \rc $ is a $c$-matching of $C_u \cup C_{v'}$ thus
\begin{align}
\nu(C_u \cup D_{v'}) - x(C_u \cup D_{v'}) &\geq  w \lb M_{C_u} \cup M_{D_{v'}} \cup \lc uv' \rc  \rb - x \lb C_u \cup C_{v'} \rb \notag \\
&= \lb w(M_{C_u})- x(C_u) \rb + \lb w(M_{ D_{v'}})- x( D_{v'})\rb + w_{uv'}  \notag \\ 
&= - z_{uv} - z_{v'u'}+w_{uv'} \quad \quad  \text{ applying  \eqref{comp} to $C_u$ , $D_{v'}$} \notag \\
&= -z_{uv} + \alpha_u \quad \quad \quad \quad  \quad \quad \text{ by choice of $v'$ and $u'$.} \notag 
\end{align}

\begin{figure}
\begin{center}
\includegraphics[width=6.6cm]{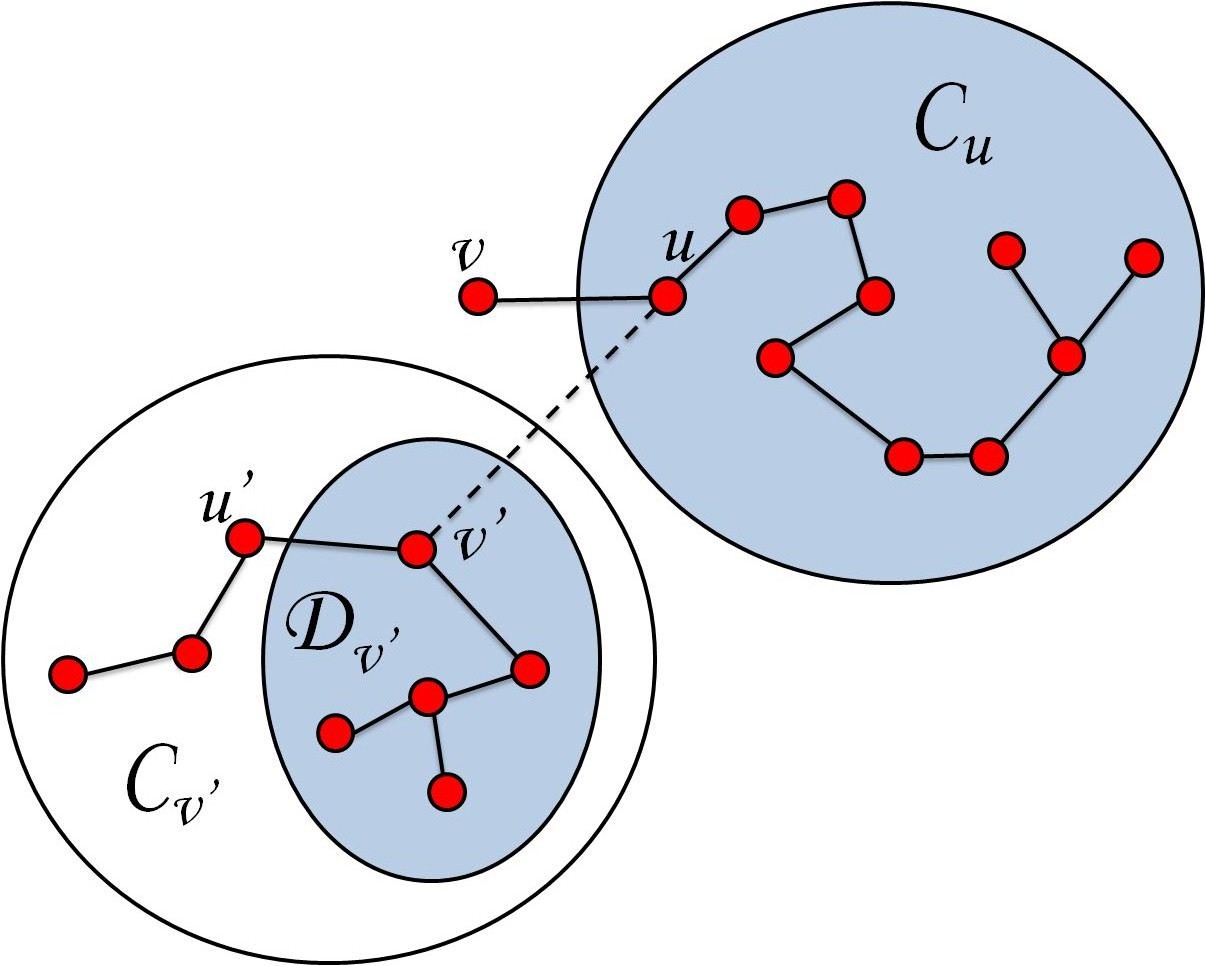}
\caption{Case 4 in the proof Theorem 1}
\end{center}
\end{figure}

Hence in all cases we have $\nu(T) -x(T) \geq - z_{uv} + \alpha_u$ as required. This completes the proof of the theorem. \qed
\end{proof}

We note that all network bargaining games studied in \cite{bateni2010cooperative} satisfy conditions $(1)$ and $(2)$ of Theorem \ref{thm: 1}. In addition to these, Theorem \ref{thm: 1} also covers the case of network bargaining games where the underlying graph is a tree, but the vertices are allowed to have arbitrary capacities.  Hence starting with a maximum weight $c$-matching $M$ we can use the polynomial time algorithm of \cite{faigle1998efficient} to compute a point in the intersection of the core and prekernel for these games, from which  we can obtain a corresponding solution $(M,z)$. Then using Theorem \ref{thm: 1} we know that $(M,z)$ will be balanced.

\section{Balanced solutions via reduction to the unit capacity games}

While we were able to generalize the class of network bargaining games for which balanced solutions can be obtained by computing a point in the intersection of the core and prekernel, we were not able to apply this technique to all network bargaining games. In this section we show that balanced solutions can be obtained to any network bargaining game $(G,w,c)$ by a reduction to a unit capacity game defined on an auxiliary graph.\\

\subsection{Construction of the instance $(G',w')$ and  matching $M'$}\label{const}

Suppose we are given an instance $\lb G, w,c \rb$  of the network bargaining game
together with a $c$-matching $M$ of $G$. We describe below  how to obtain 
an instance  $(G',w')$  of the unit capacity game together with a matching $M$ of $G$.

Construction: $[ \lb G, w,c \rb, M ] \rightarrow [(G',w'),M']$
\begin{enumerate}
\item for each $u \in V(G)$: fix an arbitrary labelling $\sigma_u : \lc v: uv \in M \rc \rightarrow \lc 1, \cdots, 
d(u) \rc$ and create $c_u$ copies $u_1, \cdots, u_{c_u}$ in $V(G')$.
\item for each $uv$ in $E(G)\cap M$: add the edge  $u_{\sigma_v(u)}v_{\sigma_u(v)}$ to $E(G') \cap M'$ and set its weight 
to $w_{uv}$.
\item for each edge $uv \in E(G) \backslash M$:   add all edges  $u_{i}v_{j}$ to $E(G')$   for all $i \in [c_u]$ and $j \in [c_v]$, and set all 
their weights to $w_{uv}$. 
\end{enumerate}

\text{}\\
\textbf{Example 1:} Consider the instance depicted on the left hand side of the figure below. The solid edges are in the $c$-matching $M$ and the dotted edges are in $E \backslash M$. Node $u$ has capacity four, nodes $x$ and $y$ have capacity two and all other nodes have capacity one. All edges have unit weight.

\begin{center}
\includegraphics[width=0.8 \textwidth]{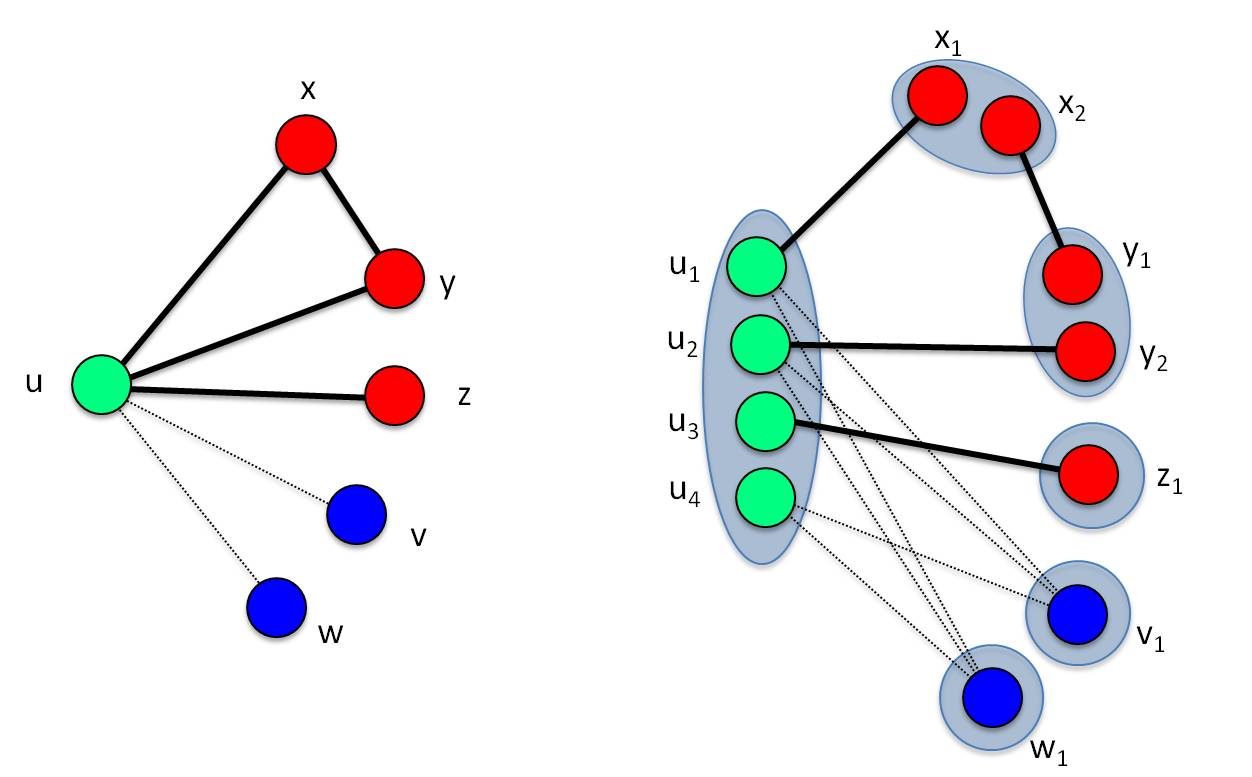}
\end{center}

We make four copies of $u$ in $G'$, two copies of $x$ and $y$, and one copy of every other node.  Each edge in $M$ corresponds to a unique edge in $M'$. For the edges $uv$ and $uw$ which are not in $M$, we connect every copy of $u$ to every copy of $v$ and $w$ with edges in $E(G') \backslash M'$. The resulting graph is on the right.

\subsection{Mapping between the two solution sets}

Suppose we are given an instance of the network bargaining game $(G,w,c)$ with a $c$-matching $M$. Let $[(G',w'),M']$ be obtained using the construction given in  section \ref{const}. Note that $M$ and $M'$ have the same number of edges and each edge  $uv \in M$ is mapped to the unique edge $u_iv_j \in M'$ where $i = \sigma_v(u)$ and $j = \sigma_u(v)$. This allows us to go back and forth  between solutions on $M$ and $M'$ by dividing the weight of each edge in the same way as its corresponding pair.

We define the two solution sets:
\begin{align}
\mathcal{X} &:= \lc x \in \RR^{|V(G')|}: \text{ $(M',x)$ is a solution to  $(G',w')$} \rc \notag \\
\mathcal{Z} &:= \lc z \in \RR^{2|E(G)|} : \text{$(M,z)$ is a solution to  $(G,w,c)$} \rc. \notag
\end{align}

And the two mappings:
\begin{enumerate}
\item $\phi: \mathcal{X} \rightarrow \mathcal{Z}$

For all $uv \in E$ define
\begin{align}
\lb \phi(x)_{uv}, \phi(x)_{vu} \rb := \begin{cases} \lb x_{u_{\sigma_v(u)}}, x_{v_{\sigma_u(v)}} \rb & \text{ if $uv \in M$,} \\
\lb 0, 0 \rb &\ \text{otherwise.} \end{cases} \notag
\end{align}

\item $\phi^{-1}: \mathcal{Z} \rightarrow \mathcal{X}$. 

For all $u_i \in V(G')$ define
\begin{align}
\phi^{-1}(z)_{u_i} := \begin{cases} z_{uv} & \text{ if $i = \sigma_u(v)$,} \\
0 &\ \text{otherwise.} \end{cases} \notag
\end{align}
\end{enumerate}

Note that $z=\phi(x)$ if and only if $x=\phi^{-1}(z)$. The following lemma shows that the mapping given by the function $\phi$ and its inverse $\phi^{-1}$ defines a bijection between the $\mathcal{X}$ and $\mathcal{Z}$

\begin{lemma}\label{lem: 3} 
\begin{enumerate}
\item If $x \in \mathcal{X}$ and $z = \phi(x)$, then $z \in \mathcal{Z}$.
\item  If  $z \in \mathcal{Z}$ and $x = \phi^{-1}(z)$ then  $x \in \mathcal{X}$.
\end{enumerate}
\end{lemma}

\begin{proof}
Let  $x \in \mathcal{X}$ and $z = \phi(x)$. We show that $z \in \mathcal{Z}$. Take $uv \in E(G) \cap M$. Suppose $i= \sigma_v(u)$ and $j = \sigma_u(v)$. 
Then from the construction of $G'$ and $M'$ it follows that  $u_iv_j \in M'$. We have
\begin{align}
z_{uv} + z_{vu} &= x_{u_i} +  x_{v_j} && \text{from the definition of $\phi(x)$} \notag \\
&= w_{u_iu_j} && \text{since $(M',x)$ is a solution } \notag \\
&= w_{uv} && \text{from the construction of $(G',w')$}. \notag
\end{align}
Furthermore if $uv \in E(G) \backslash M$ then from the definition of $\phi(x)$ we have $z_{uv} =z_{vu} = 0$. 

Now let  $z \in \mathcal{Z}$ and $x = \phi^{-1}(z)$. We show that  $x \in \mathcal{X}$. Take   $u_iv_j \in E(G') \cap M'$. From the construction of $G'$ and $M'$ there 
must exist an edge $uv \in E(G) \cap M$ such that  $i= \sigma_v(u)$ and $j = \sigma_u(v)$.  We 
have
\begin{align}
x_{u_i} + x_{v_j} &= z_{uv} +  z_{vu} && \text{from the definition of $\phi^{-1}(z)$} \notag \\
&= w_{uv} && \text{since $(M,z)$ is a solution } \notag \\
&= w_{u_iv_j} && \text{from the construction of $(G',w')$}. \notag
\end{align}
Furthermore if $u_i$ is uncovered in $M'$ then $x_{u_i} = 0$ by definition. \qed
\end{proof}

From now on we write $(M,z) \sim (M',x)$ whenever $z = \phi(x)$ or  equivalently $x = \phi^{-1}(z)$.  The next  lemma is the key step in showing that certain properties of a solution are preserved under our mapping.

\begin{lemma}\label{lem: key} Let $(G,w,c)$ be an instance of the network bargaining game and $M$ a  $c$-matching on $G$. Suppose the auxiliary instance $(G',w')$ and the matching $M'$ were obtained using the  construction given in section \ref{const}. Let $(M,z)$ be a solution to $(G,w,c)$ and $(M',x')$  
a solution to $(G',w')$ such that $(M,z) \sim (M',x)$. Then for any $u \in V(G)$ and any $i 
\in [d_u]$ we have 
\begin{align}
\alpha_u(M,z) = \alpha_{u_i}\lb M', x \rb.\notag
\end{align}
\end{lemma}

\begin{proof}
We first show that  $\alpha_u(M,z) \leq \alpha_{u_i}\lb M', x \rb$. We may 
assume that $ \alpha_u(M,z) > 0$ since otherwise there is nothing to show. Let $v$ be 
vertex $u$'s best outside option in $(M,z)$. That is $uv \in E(G) \backslash M$ and
\begin{align}
\alpha_u(M,z) &=  w_{uv}-\textbf{1}_{[d_v= c_v]} \min_{vw \in M}  z_{vw}  \notag .
\end{align}

Since $uv \in E(G) \backslash M$ we have $u_i v_j \in E(G') \backslash M'$ for all $j \in [c_v]$.  We have two cases:
\begin{enumerate}
\item $v$ is not saturated in $M$. Then the vertex $v_{d_v+1}$ is in $V(G')$ and it is not covered by $M'$.  Since 
$u_iv_{d_v+1} \in E(G') \backslash M'$ we have
\begin{align}
\alpha_{u_i}\lb M', x \rb &\geq w_{u_iv_{d_v+1}} - x_{v_{d_v+1}} \notag \\
& = w_{u_iv_{d_v+1}} && \text{$v_{d_v+1}$ is not covered by $M'$ so $x_{v_{d_v+1}} = 0$} \notag \\
&= w_{uv} && \text{from the definition of $(G',w')$} \notag \\
&= \alpha_u(M,z)  && \text{by choice of $v$.} \notag
\end{align}
\item  $v$ is saturated in $M$. Let $w = \arg \min_{vw \in M} z_{vw}$. Suppose that $j = \sigma_w(v)$. Then 
$v_j$ is covered in $M'$ and $x_{v_j} = z_{vw}$. Since $u_iv_{j} \in E(G') \backslash M'$ we have
\begin{align}
\alpha_{u_i}\lb M', x \rb &\geq w_{u_iv_{j}} - x_{v_j} \notag \\
& = w_{uv} - z_{vw}   && \text{from the definition of $(G',w')$}\notag \\
&= \alpha_u(M,z) && \text{by choice of $v$.} \notag
\end{align}
\end{enumerate}

We now show that  $\alpha_u(M,z) \geq \alpha_{u_i}\lb M', x \rb$. We may 
assume that $ \alpha_{u_i}\lb M', x \rb> 0$. Let $v_j$ be vertex $u_i$'s best outside 
option in $(M',x)$. That is,  $v_j \in V(G')$ such that $u_iv_j \in E(G') \backslash M'$ and
\begin{align}
\alpha_{u_i}(M',x) &=  w_{u_iv_j}-x_{v_j}\notag .
\end{align}
Since $u_iv_j \in E(G') \backslash M'$ we must have $uv \in E(G) \backslash M$. Again, we have two cases:
\begin{enumerate}
\item $v_j$ is not covered in $M'$. Then the vertex $v$ is not saturated in $M$ and 
\begin{align}
\alpha_u(M,z)  &\geq w_{uv} \notag \\
& = w_{u_iv_{j}} && \text{from the definition of $(G',w')$} \notag \\ 
&= w_{u_jv_j} - x_{v_j} && \text{$v_j$ is not covered in $M'$ so $x_{v_{j}} = 0$} \notag \\
&=  \alpha_{u_i}\lb M', x \rb && \text{by choice of $v_j$.} \notag
\end{align}
\item $v_j$ is covered in $M$. Then there exists $w \in V(G)$ such that $vw \in E(G) \cap M$ and $j = \sigma_{w}(v)$ .  We have
\begin{align}
\alpha_u(M,z)  &\geq w_{uv} - z_{vw} \notag \\
&= w_{u_iv_{j}}  - x_{v_j} && \text{from the definition of $(G',w')$} \notag \\ 
&=  \alpha_{u_i}\lb M', x \rb && \text{by choice of $v_j$.} \notag
\end{align}
\end{enumerate}
\qed
\end{proof}

Using Lemma \ref{lem: key} we can now prove that  stability and balance are preserved when mapping between solutions of the network bargaining game and the corresponding unit capacity game of the auxiliary instance.

\begin{theorem}\label{thm: main} Let $(G,w,c)$ be an instance of the network bargaining game and $M$ a  $c$-matching on $G$. Suppose the auxiliary instance $(G',w')$ and the matching $M'$ were obtained using the  construction given in section \ref{const}. Let $(M,z)$ be a solution to $(G,w,c)$ and $(M',x')$  
a solution to $(G',w')$ such that $(M,z) \sim (M',x)$.  Then:
\begin{enumerate}
\item  $(M,z)$ is stable if and only if $(M',x)$ is stable.
\item $(M,z)$ is balanced if and only if $(M',x)$ is balanced.
\end{enumerate}
\end{theorem}

\begin{proof}  Let $uv \in M$. Suppose that $i= \sigma_v(u)$. Then $z_{uv} = x_{u_i} $ and 
using Lemma \ref{lem: key} we have
\begin{align}
z_{uv} &\geq \alpha_u(M,z) \quad \text{if and only if} \quad  x_{u_i} \geq  
\alpha_{u_{\sigma_v(u)}}\lb M', x \rb \notag . 
\end{align}
It remains to show that if $(M',x)$ is stable then $\alpha_u(M,z) = 0$ for any 
unsaturated vertices $u$ of $G$. Suppose $u$ is such a vertex. Then the vertex 
$u_{d_u+1}$ is not covered in $M'$ and therefore $x'_{d_u+1} = 0$. If $(M',x)$ is stable 
then $\alpha_{u_{d_u+1}} = 0$ and by Lemma \ref{lem: key} we have $\alpha_u(M,z) = 0$ as desired. 
This completes the proof of the first statement. To prove the second statement let  $uv \in M$ and suppose that $i= 
\sigma_v(u)$ and $j=\sigma_u(v)$. Then $z_{uv} = x_{u_i}$, $z_{vu} = x_{v_j}$ and by 
 Lemma \ref{lem: key} we have:
\begin{align}
z_{uv} - \alpha_u(M,z) =  z_{vu} - \alpha_v(M,z)\quad \Leftrightarrow \quad x_{u_i} - 
\alpha_{u_i}\lb M', x \rb = x_{v_j} - \alpha_{v_j}\lb M', x \rb. \notag
\end{align} 
This completes the proof. \qed
\end{proof}

\subsection{Algorithm for computing balanced solutions}\label{sec1}

Using Theorem \ref{thm: main} we have the following algorithm for finding a balanced solution to the network bargaining game $(G,w,c)$:

\begin{enumerate} 
\item Find a maximum $c$-matching $M$ in $G$.
\item Obtain unit capacity game $(G',w')$ with matching $M'$ using the construction from section \ref{const}.
\item Find a balanced solution $x$ on the matching $M'$ in $G'$.
\item Set $z= \phi(x)$ and return $(M,z)$.
\end{enumerate} 

We note that step 3 of the algorithm can be implementing using the existing polynomial time algorithm of Kleinberg and Tardos \cite{kleinberg2008balanced}. Given any instance of a network bargaining game with unit capacities together with a maximum weight matching, their algorithm returns a balanced solution on the given matching, whenever one exists.

\subsection{Existence of balanced solutions}\label{sec2}

Using Theorem \ref{thm: main} we know that stable solutions of the original problem map to stable solutions of the matching problem and viceversa. Since any stable solution must occur on a $c$-matching, respectively matching, of maximum weight we have the following corollary

\begin{corollary} Let $(G,w,c)$ be an instance of the network bargaining game and $M$ a  $c$-matching 
on $G$. Suppose the auxiliary instance $(G',w')$ and the matching $M'$ were obtained 
using the given construction. Then 
\begin{enumerate}
\item $M$ is a maximum weight $c$-matching for $(G,w,c)$ if and only if $M'$ is a maximum weight 
matching for $(G',w')$.
\item There exists a balanced solution for $(G,w,c)$ on the $c$-matching $M$ if and only if there exists a balanced 
solution for $(G',w')$ on the matching $M'$.
\end{enumerate}
\end{corollary}

It was previously shown in \cite{kleinberg2008balanced} that a unit capacity game possesses a balanced solution if and only if it has a stable solution, which in turn happens if and only if the linear program for the maximum weight matching of the underlying graph has an integral optimal solution. For the case of network bargaining game with general capacities, \cite{kanoria2009natural} have shown that a stable solution exists if and only if the linear program for the maximum weight $c$-matching of the underlying graph has an integral optimal solution. In terms of existence of balanced solutions, they only  obtain a partial characterization by proving  that if this integral optimum is unique, then a balanced solution is guaranteed to exist. Our results imply the following full characterization for the existence of balanced solutions, thus extending the results of \cite{kanoria2009natural}:

\begin{theorem} An instance $(G,w,c)$ of the network bargaining game has a balanced solution if and only if it has a stable one.
\end{theorem}

\end{document}